\documentclass[12pt]{amsart}
\usepackage{a4wide}
\usepackage[dvips]{epsfig}
\usepackage{amscd}
\usepackage{amssymb}
\usepackage{amsthm}
\usepackage{amsmath}
\usepackage{latexsym}
\usepackage[linktocpage,
                  menucolor=blue,
                  linkcolor=blue,
                  citecolor=red,
                  colorlinks=true]{hyperref}

\theoremstyle{plain}
\newtheorem{theorem}{Theorem}[section]
\theoremstyle{plain}

\theoremstyle{definition}

\newtheorem{example}{Example}[section]

{%
\setcounter{enumi}{0}

\begin{enumerate}}%
{\end{enumerate} }
{%
\setcounter{enumi}{0}

\begin{enumerate}}%
{\end{enumerate} }

\numberwithin{equation}{section}
 \allowdisplaybreaks

\title[Optimal portfolios, consumption and  insurance]
 {On the optimal investment-consumption and life insurance selection problem with an external stochastic factor }

\date{\today}

\begin{document}

\author{Calisto Guambe}
\address{Department of Mathematics and Applied Mathematics, University of Pretoria, 0002, South Africa}
\address{Department of Mathematics and Informatics, Eduardo Mondlane University, 257, Mozambique}

\email{calistoguambe@yahoo.com.br}

\author{ Rodwell Kufakunesu }

\address{Department of Mathematics and Applied Mathematics
, University of Pretoria, 0002, South Africa}

\email{rodwell.kufakunesu@up.ac.za}


\begin{abstract} In this paper, we study a stochastic optimal control problem with stochastic volatility. We prove the sufficient and necessary maximum principle for the proposed problem. Then we apply the results to solve an investment, consumption and life insurance problem with stochastic volatility, that is, we consider a wage earner investing in one risk-free asset and one risky asset described by a jump-diffusion process and has to decide concerning consumption and life insurance purchase. We assume that the life insurance for the wage earner is bought from a market composed of $M>1$ life insurance companies offering pairwise distinct life insurance contracts. The goal is to maximize the expected utilities derived from the consumption, the legacy in the case of a premature death and the investor's terminal wealth.
\end{abstract}

\maketitle
\section{Introduction}
The problem of a wage earner who wants to invest and protect his dependents for a possible premature death has gained much interest in recent times. Since the research paper on portfolio optimization and life insurance purchase by Richard \cite{Richard} appeared,  a number of works in this direction have been reported in the literature. For instance, Pliska and Ye \cite{Pliska} studied an optimal consumption and life insurance contract for a problem described by a risk-free asset. Duarte {\it et al.} \cite{Duarte} considered a problem of a wage earner who invests and buys a life insurance in a financial market with $n$ diffusion risky shares. Similar works include (Guambe and Kufakunesu \cite{guambe}, Huang {\it et al.} \cite{Huang}, Liang and Guo \cite{liang}, Shen and Wei \cite{Shen}, among others). In all the above-mentioned papers, a single life insurance contract was considered.

Recently, Mousa {\it et al.} \cite{mousa}, extended Duarte {\it et al.} \cite{Duarte} to consider a wage earner who buys life insurance contracts from $M>1$ life insurance companies. Each insurance company offers pairwise distinct contracts. This allows the wage earner to compare the premiums insurance ratio of the companies and buy the amount of life insurance from the one offering the smallest premium-payout ratio at each time. Using a dynamic programming approach, they solved the optimal investment, consumption and life insurance contracts in a financial market comprised by one risk-free asset and $n$ risky shares driven by diffusion processes. In this paper, we extend their work to a jump-diffusion setup with stochastic volatility. This extension is motivated by the following reasons: First, the existence of high frequency data on the empirical studies carried out by Cont \cite{cont}, Tankov \cite{tankov} and references therein, have shown that the analysis of price evolution reveals some sudden changes that cannot be explained by models driven by diffusion processes. Another reason is related to the presence of volatility clustering in the distribution of the risky share process, i.e., large changes
in prices are often followed by large changes and small changes tend to be followed by
small changes.

To enable a full capture of these and other aspects, we consider a jump diffusion model with stochastic volatility similar to that in Mnif \cite{Mnif}. Using Dynamic programming approach, Mnif \cite{Mnif} proved the existence of a smooth solution of a semi-linear integro-Hamilton-Jacobi-Bellman (HJB) for the exponential utility function. Zeghal and Mnif \cite{zeghal} considered the same problem for power utility case. Under some particular assumptions, they also derived the backward stochastic differential equation (BSDE) associated with the semi-linear HJB. The drawback of the dynamic programming approach is that it requires the system to be Markovian. To overcome this limitation, a maximum principle approach is proposed to solve this stochastic volatility jump-diffusion problem. This approach allows to solve this problem in a more general setting. We prove a sufficient and necessary maximum principle in a general stochastic volatility problem. Then we apply this framework to solve the wage earner investment, consumption and life insurance problem described earlier. In the literature, the maximum principle approach has been widely reported, see, for instance, Framstad {\it et. al.} \cite{framstard}, \O ksendal and Sulem \cite{Oksendal}, An and \O ksendal \cite{An-oksendal}, Pamen \cite{pamen2015}, Pamen and Momeya \cite{pamen2017}, among others. The main contribution of this paper is the use of factor model in the investment, consumption and life insurance contract problem as well as the inclusion of jumps in the modeling framework.

The rest of the paper is organized as follows: in Section 2, we introduce our control problem and state the sufficient and necessary maximum principle for a stochastic control problem with stochastic volatility, where the proofs are given in the Appendix. In Section 3,  we give the characterization of the optimal strategies for  an investment, consumption and life insurance problem applying the results of Theorem \ref{sufficientmp}. Finally, we consider an example of a linear pure jump stochastic volatility model of Ornstein-Uhlenbeck type and derive an explicit optimal portfolio.

\section{Maximum principle for stochastic optimal control problem with stochastic volatility}

Let $T<\infty$ be a finite time horizon investment period, which can be viewed as a retirement time of an investor. Consider two independent Brownian motions $\{W_1(t);W_2(t),\, 0\leq t\leq T\}$ associated to the complete filtered probability space $(\Omega^W,\mathcal{F}^W,\{\mathcal{F}_t^W\},\mathbb{P}^W)$. Furthermore, we consider a Poisson process $N$ independent of  $W_1$ and $W_2$, associated with the complete filtered probability space $(\Omega^N,\mathcal{F}^N,\{\mathcal{F}_t^N\},\mathbb{P}^N)$ with the intensity measure $dt\times d\nu(z)$, where $\nu$ is the $\sigma$-finite Borel measure on $\mathbb{R}\setminus\{0\}$. A $\mathbb{P}^N$-martingale compensated Poisson random measure is given by:
\begin{equation*}
    \tilde{N}(dt,dz):=N(dt,dz)-\nu(dz)dt\,.
\end{equation*}

We define the product space:
$$(\Omega,
\mathcal{F},\{\mathcal{F}_t\}_{0\leq t\leq T},\mathbb{P}):=(\Omega^W\otimes\Omega^N,\mathcal{F}^W\otimes\mathcal{F}^N, \{\mathcal{F}^W\otimes\mathcal{F}^N\},\mathbb{P}^W\otimes\mathbb{P}^N)
$$ where
$\{\mathcal{F}_t\}_{t\in[0,T]}$ is a filtration satisfying the usual conditions.

Suppose that the dynamics of the state process is given by the following stochastic differential equation (SDE)
\begin{eqnarray}\label{stateprocess}
  dX(t) &=& b(t,X(t),Y(t),\pi(t))dt+ \sigma(t,X(t),Y(t),\pi(t))dW_1(t)  \\ \nonumber
   && +\beta(t,X(t),Y(t),\pi(t))dW_2(t) +\int_{\mathbb{R}}\gamma(t,X(t),Y(t),\pi(t),z)\tilde{N}(dt,dz)\,; \\ \nonumber
  X(0) &=& x\in\mathbb{R}\,,
\end{eqnarray}
where the external economic factor $Y$ is given by
\begin{equation}\label{externalgeneral}
    dY(t)=\varphi(Y(t))dt+\phi(Y(t))dW_2(t)\,.
\end{equation}

We assume that the functions $b,\sigma,\beta:[0,T]\times\mathbb{R}\times\mathbb{R}\times\mathcal{A}\rightarrow\mathbb{R}$; $\gamma:[0,T]\times\mathbb{R}\times\mathbb{R}\times\mathcal{A}\times\mathbb{R}\rightarrow\mathbb{R}$; $\varphi,\phi:\mathbb{R}\rightarrow\mathbb{R}$ are given predictable processes, such that \eqref{stateprocess} and \eqref{externalgeneral} are well defined and \eqref{stateprocess} has a unique solution for each $\pi\in\mathcal{A}$. Here, $\mathcal{A}$ is a given closed set in $\mathbb{R}$.

Let $f:[0,T]\times\mathbb{R}\times\mathbb{R}\times\mathcal{A}\rightarrow\mathbb{R}$ be a continuous function and $g:\mathbb{R}\times\mathbb{R}\rightarrow\mathbb{R}$ a concave function. We define the performance criterion by
\begin{equation}\label{functionalgeneral}
    \mathcal{J}(\pi)=\mathbb{E}\Bigl[\int_0^Tf(t,X(t),Y(t),\pi(t))dt+g(X(T),Y(T))\Bigl]\,.
\end{equation}
We say that $\pi\in\mathcal{A}$ is an admissible strategy if \eqref{stateprocess} has a unique strong solution and
$$
\mathbb{E}\Bigl[\int_0^T|f(t,X(t),Y(t),\pi(t))|dt+|g(X(T),Y(T))|\Bigl]<\infty\,.
$$
The main problem is to find $\pi^*\in\mathcal{A}$ such that
$$
\mathcal{J}(\pi^*)=\sup_{\pi\in\mathcal{A}}\mathcal{J}(\pi)\,.
$$
The control $\pi^*$ is called an optimal control if it exists.

In order to solve this stochastic optimal control problem with stochastic volatility, we use the so called maximum principle approach. The beauty of this method is that it solves a stochastic control problem in a more general situation, that is, for both Markovian and non-Markovian cases. For the Markovian case, this problem has been solved using dynamic programming approach by Mnif \cite{Mnif}. Our approach may be considered as an extension of the maximum approach in Framstad {\it et. al.} \cite{framstard} to the stochastic volatility case.

We define the Hamiltonian $\mathcal{H}:[0,T]\times \mathbb{R}\times\mathbb{R}\times\mathbb{A}\times\mathbb{R}\times\mathbb{R}\times\mathbb{R}\times\mathbb{R} \times\mathbb{R}\rightarrow\mathbb{R}$ by:

\begin{eqnarray}\label{hamiltoniangeneral}
  && \mathcal{H}(t,X(t),Y(t),\pi(t),A_1(t),A_2(t),B_1(t),B_2(t),D_1(t,\cdot))  \\ \nonumber
  &=& f(t,X(t),Y(t),\pi(t))+b(t,X(t),Y(t),\pi(t))A_1(t) +\varphi(Y(t))A_2(t) \\ \nonumber
   && + \sigma(t,X(t),Y(t),\pi(t))B_1(t) +\beta(t,X(t),Y(t),\pi(t))B_2(t)+\phi(Y(t))B_3(t) \\ \nonumber
   && +\int_{\mathbb{R}}\gamma(t,X(t),Y(t),\pi(t),z)D_1(t,z)\nu(dz)\,,
\end{eqnarray}
provided that the integral in \eqref{hamiltoniangeneral} converges.  From now on, we assume that the Hamiltonian $\mathcal{H}$ is continuously differentiable w.r.t. $x$ and $y$. Then, the adjoint equations corresponding to the admissible strategy $\pi\in\mathcal{A}$ are given by the following backward stochastic differential equations (BSDEs)
\begin{eqnarray}\nonumber
  dA_1(t) &=& -\frac{\partial\mathcal{H}}{\partial x}(t,X(t),Y(t),\pi(t),A_1(t),A_2(t),B_1(t),B_2(t),D_1(t,\cdot))dt  \\ \label{adjointg1}
   && +B_1(t)dW_1(t) +B_2(t)dW_2(t)+\int_{\mathbb{R}}D_1(t,z)\tilde{N}(dt,dz)\,, \\
  A_1(T) &=& \frac{\partial g}{\partial x}(X(T),Y(T))
\end{eqnarray}
and
\begin{eqnarray}\nonumber
  dA_2(t) &=& -\frac{\partial\mathcal{H}}{\partial y}(t,X(t),Y(t),\pi(t),A_1(t),A_2(t),B_1(t),B_2(t),D_1(t,\cdot))dt  \\ \label{adjointg2}
   && +B_3(t)dW_1(t) +B_4(t)dW_2(t)+\int_{\mathbb{R}}D_2(t,z)\tilde{N}(dt,dz)\,, \\
  A_2(T) &=& \frac{\partial g}{\partial y}(X(T),Y(T))\,.
\end{eqnarray}

The verification theorem associated to our problem is stated as follows:

\begin{theorem}\label{sufficientmp}
$\rm{ (Sufficient \ maximum \ principle)}$
Let $\pi^*\in\mathcal{A}$ with the corresponding wealth process $X^*$. Suppose that the pairs $(A_1^*(t),B_1^*(t),B_2^*(t),D_1^*(t,z))$ and\\ $(A_2^*(t),B_3^*(t),B_4^*(t),D_2^*(t,z))$ are the solutions of the adjoint equations \eqref{adjointg1} and \eqref{adjointg2}, respectively. Moreover, suppose that the following inequalities hold:
\begin{itemize}
  \item[(i)] The function $(x,y)\rightarrow g(x,y)$ is concave;
  \item[(ii)] The function $\mathcal{H}(t)=\sup_{\pi\in\mathcal{A}}\mathcal{H}(t,X(t),Y(t),\pi,A_1^*(t),A_2^*(t),B_1^*(t),B_2^*(t),D_1^*(t,z))$ is concave and
  $$
\mathcal{H}^*(t,X,Y,\pi^*,A_1^*,A_2^*,B_1^*,B_2^*,D_1^*)=\sup_{(\pi,c,p)\in\mathcal{A}} \mathcal{H}(t,X,Y,\pi,A_1^*,A_2^*,B_1^*,B_2^*,D_1^*)\,.
$$

\end{itemize}
Furthermore, we assume the following:
$$
\mathbb{E}\Bigl[\int_0^T(X^*(t))^2\Bigl((B_1^*(t))^2+(B_2^*(t))^2+\int_{\mathbb{R}}(D_1^*(t,z))^2\nu(dz)\Bigl)dt\Bigl]<\infty\,;
$$
$$
\mathbb{E}\Bigl[\int_0^T(Y(t))^2\Bigl((B_3^*(t))^2+(B_4^*(t))^2+\int_{\mathbb{R}}(D_2^*(t,z))^2\nu(dz)\Bigl)dt\Bigl]<\infty\,;
$$

\begin{eqnarray*}
\mathbb{E}\Bigl[\int_0^T\Bigl\{(A_1^*(t))^2\Bigl((\sigma(t,X(t),Y(t),\pi^*(t)))^2+(\beta(t,X(t),Y(t),\pi^*(t)))^2 && \\
 +\int_{\mathbb{R}}(\gamma(t,X(t),Y(t),\pi^*(t),z))^2\nu(dz)\Bigl)+(A_2^*(t))^2(\phi(Y(t)))^2\Bigl]dt\Bigl] &<& \infty\,.
\end{eqnarray*}

Then, $\pi^*\in\mathcal{A}$ is an optimal strategy with the corresponding optimal state process $X^*$.
\end{theorem}

\begin{proof}
See Appendix.

\end{proof}

Note that the sufficient maximum principle presented in Theorem \ref{sufficientmp} is based on the concavity of the Hamiltonian, however, this condition does not hold in many concrete situations. Below, we relax this condition and state the necessary maximum principle for our control problem. Thus, we further consider the following assumptions.
\begin{itemize}
  \item For all $s\in[0,T]$ and all bounded $\mathcal{F}_s$-measurable random variable $\alpha(\omega)$, the control $\xi(t):=\mathbf{1}_{[s,T]}(t)\alpha(\omega)$ belongs to the admissible strategy $\mathcal{A}$.
  \item For all $\pi,\zeta\in\mathcal{A}$, with $\zeta$ bounded, there exists $\epsilon>0$ such that the control $\pi(t)+\ell\zeta(t)\in\mathcal{A}$, for all $\ell\in(-\epsilon;\epsilon)$.
  \item We define the derivative processes
  $$
  x_1(t):=\frac{d}{d\ell}X^{\pi+\ell\zeta}(t)\Bigl|_{\ell=0} \ \ \ \ \rm{and} \ \ \ \ y_1(t):=\frac{d}{d\ell}Y^{\pi+\ell\zeta}(t)\Bigl|_{\ell=0}\,.
  $$
  Then, for all $\pi,\zeta\in\mathcal{A}$, with $\zeta$ bounded, the above derivatives exist and belong to $L^2([0,T]\times\Omega)$, and \eqref{stateprocess} and \eqref{externalgeneral},
  \begin{eqnarray*}
    dx_1(t) &=& x_1(t)\Bigl[\frac{\partial b}{\partial x}(t)dt +\frac{\partial \sigma}{\partial x}(t)dW_1(t)+ \frac{\partial \beta}{\partial x}(t)dW_2(t) +\int_{\mathbb{R}}\frac{\partial \gamma}{\partial x}(t,z)\tilde{N}(dt,dz)\Bigl] \\
     && +y_1(t)\Bigl[\frac{\partial b}{\partial y}(t)dt +\frac{\partial \sigma}{\partial y}(t)dW_1(t)+ \frac{\partial \beta}{\partial y}(t)dW_2(t) +\int_{\mathbb{R}}\frac{\partial \gamma}{\partial y}(t,z)\tilde{N}(dt,dz)\Bigl] \\
     && +\zeta(t)\Bigl[\frac{\partial b}{\partial \pi}(t)dt +\frac{\partial \sigma}{\partial \pi}(t)dW_1(t)+ \frac{\partial \beta}{\partial \pi}(t)dW_2(t) +\int_{\mathbb{R}}\frac{\partial \gamma}{\partial \pi}(t,z)\tilde{N}(dt,dz)\Bigl]\,,
  \end{eqnarray*}
  where we have used the notation $\frac{\partial b}{\partial x}(t)=\frac{\partial b}{\partial x}(t,X(t),Y(t),\pi(t))$, \\ $\frac{\partial \sigma}{\partial x}(t)=\frac{\partial \sigma}{\partial x}(t,X(t),Y(t),\pi(t))$, etc. Moreover,
  \begin{equation*}
    dy_1(t)=y_1(t)[\varphi'(Y(t))dt+\phi'(Y(t))dW_2(t)]\,.
  \end{equation*}
\end{itemize}

\begin{theorem}\label{necessarymp}
$\rm{ (Necessary \ maximum \ principle)}$ Let $\pi\in\mathcal{A}$ with corresponding solutions $X(t)$, $(A_1(t),B_1(t),B_2(t),D_1(t,\cdot))$, $(A_2(t),B_3(t),B_4(t),D_2(t.\cdot))$ of \eqref{stateprocess}, \eqref{adjointg1} and \eqref{adjointg2} respectively, and the derivative processes $x_1(t)$ and $y_1(t)$ given above. Moreover, assume the following integrability conditions:
\begin{eqnarray*}
  \mathbb{E}\Bigl\{\int_0^TA_1^2(t)\Bigl[x_1^2(t)\Bigl(\Bigl(\frac{\partial \sigma}{\partial x}(t)\Bigl)^2 +\Bigl(\frac{\partial \beta}{\partial x}(t)\Bigl)^2 +\int_{\mathbb{R}}\Bigl(\frac{\partial \gamma}{\partial x}(t,z)\Bigl)^2\nu(dz)\Bigl) &&  \\
    y_1^2(t)\Bigl(\Bigl(\frac{\partial \sigma}{\partial y}(t)\Bigl)^2 +\Bigl(\frac{\partial \beta}{\partial y}(t)\Bigl)^2 +\int_{\mathbb{R}}\Bigl(\frac{\partial \gamma}{\partial y}(t,z)\Bigl)^2\nu(dz)\Bigl) && \\
    \zeta^2(t)\Bigl(\Bigl(\frac{\partial \sigma}{\partial \pi}(t)\Bigl)^2 +\Bigl(\frac{\partial \beta}{\partial \pi}(t)\Bigl)^2 +\int_{\mathbb{R}}\Bigl(\frac{\partial \gamma}{\partial \pi}(t,z)\Bigl)^2\nu(dz)\Bigl)\Bigl]dt && \\
     +\int_0^TA_2^2(t)y_1^2(t)(\phi'(Y(t)))^2dt\Bigl\} && \\
   &<& \infty
\end{eqnarray*}
and
\begin{eqnarray*}
  \mathbb{E}\Bigl\{\int_0^Tx_1^2(t)\Bigl[B_1^2(t)+B_2^2(t)+\int_{\mathbb{R}}D_1^2(t,z)\nu(dz)\Bigl]dt &&  \\
 \int_0^Ty_1^2(t)\Bigl[B_3^2(t)+B_4^2(t)+\int_{\mathbb{R}}D_2^2(t,z)\nu(dz)\Bigl]dt\Bigl\}  &<& \infty\,.
\end{eqnarray*}

Then following are equivalent
\begin{enumerate}
  \item $\frac{d}{d\ell}\mathcal{J}(\pi+\ell\zeta)\Bigl|_{\ell=0}=0$ for all bounded $\zeta\in\mathcal{A}$;
  \item $\frac{d\mathcal{H}}{d\pi}(t,X^*(t),Y(t),\pi^*(t),A_1^*(t),A_2^*(t),B_1^*(t),B_2^*(t),D_1^*(t,z))=0$, for all $t\in[0,T]$.
\end{enumerate}
\end{theorem}
\begin{proof}
See Appendix.
\end{proof}

\section{Application to optimal investment-consumption and life insurance selection problem}

We consider a financial market consisting of one risk-free asset $(B(t))_{0\leq t\leq T}$ and one risky  asset $(S(t))_{0\leq t\leq T}$. Their respective prices are given by the following  SDE:

\begin{eqnarray}
  dB(t) &=& r(t)B(t)dt\,, \ \ B(0)=1\,, \label{risk-free} \\
  dS(t) &=& S(t)\Bigl[\alpha(t,Y(t))dt+\beta(t,Y(t))dW_1(t)+\sigma(t,Y(t))dW_2(t) \\ \nonumber
  && \ \ \ \ \  \ \  +\int_{\mathbb{R}}\gamma(t,Y(t),z)\widetilde{N}(dt,dz)\Bigl]\,, \ \ \ S(0)=s>0\,, \label{risky}
\end{eqnarray}
where $Y$ is a continuous time economic external factor governed by
\begin{equation}\label{external}
    dY(t)=g(Y(t))dt+dW_1(t)\,.
\end{equation}

Here, the associated parameters in the model satisfy the following assumptions:
\begin{itemize}
  \item[$(\textbf{A1})$] The interest rate $r(t)$ is positive, deterministic and integrable for all $t\in[0,T]$. The mean rate of return $\alpha$, the volatilities $\beta\,,\sigma$ and the dispersion rate $\gamma>-1$, are $\mathbb{R}$-valued functions are assumed to be continuously differentiable functions $(\in\mathcal{C}^1)$ and bounded.  Note that, by the continuity of $Y$, the process $S$ in \eqref{risky} is well defined on $[0,T]$. We also assume the following integrability condition:
      $$
      \mathbb{E}\left[\int_0^T(\beta^2(t,y)+ \sigma^2(t,y) +\int_{\mathbb{R}\setminus\{0\}}|\gamma(t,y,z)|^2\nu(dz))dt\right]<\infty\,.
      $$

\end{itemize}
Suppose that $g\in\mathcal{C}^1(\mathbb{R})$ with the first derivative bounded, i.e., $|g'(y)|\leq K$ and satisfy  a Lipschitz condition on the $\mathbb{R}$-valued function $g$:

\begin{itemize}
  \item[$(\textbf{A2})$] There exists a positive constant $C$ such that:
  $$ |g(y)-g(w)|\leq C|y-w|\,, \ \ \ \ \ \  y,w\in\mathbb{R}\,.$$
\end{itemize}

Consider a wage earner whose life time is a nonnegative random variable $\tau$ defined on the probability space $(\Omega,\mathcal{F},\{\mathcal{F}_t\}_{0\leq t\leq T},\mathbb{P})$. 
 The distribution function $F(t)$ of the random lifetime $\tau$ with the probability density function $f(t)$ is given by
$$
F(t):=\mathbb{P}(\tau<t)=\int_0^tf(s)ds\,.
$$
Thus, the probability that the lifetime $\tau>t$ is given by:
$$
\bar{F}(t):=\mathbb{P}(\tau\geq t\mid\mathcal{F}_t)=1-F(t)\,.
$$
We introduce the instantaneous force of mortality $\lambda(t)$ for the wage earner to be alive at time $t$. By definition, $\lambda(t)$ is given by:
\begin{eqnarray*}
  \lambda(t) &:=& \lim_{\Delta t\rightarrow0}\frac{\mathbb{P}(t\leq\tau<t+\Delta t|\tau\geq t)}{\Delta t} \\
   &=& \lim_{\Delta t\rightarrow0}\frac{\mathbb{P}(t\leq\tau<t+\Delta t)}{\Delta t\mathbb{P}(\tau\geq t)} \\
   &=&  \frac{1}{\bar{F}(t)}\lim_{\Delta t\rightarrow0}\frac{F(t+\Delta t)-F(t)}{\Delta t} \\
   &=& \frac{f(t)}{\bar{F}(t)}=-\frac{d}{dt}(\ln(\bar{F}(t)))\,.
\end{eqnarray*}
Then, the conditional survival probability of the wage earner is given by:
 \begin{equation}\label{survival}
\bar{F}(t)=\mathbb{P}(\tau>t|\mathcal{F}_t)=\exp\left(-\int_0^t\lambda(s)ds\right),
 \end{equation}
and the conditional survival probability density of the death of the wage earner by:
 \begin{equation}\label{death}
    f(t):=\lambda(t)\exp\left(-\int_0^t\lambda(s)ds\right)\,.
 \end{equation}

As in Mousa {\it et al.} \cite{mousa}, we suppose the existence of an insurance market composed of $M$ insurance companies, with each insurance company continuously offering life insurance contracts. We assume that the wage earner is paying premium insurance rate $p_n(t)$, at time $t$ for each company $n=1,2,\ldots,M$. If the wage earner dies, the insurance companies will pay $p_n(\tau)/\eta_n(\tau)$ to his/her beneficiary. Here, $\eta_n>0$ is the $n$th deterministic insurance company premium-payout ratio. Additionally, we assume that the $M$ insurance companies under consideration offer pairwise distinct contracts in the sense that $\eta_{n_1}(t)\neq\eta_{n_2}(t)$, for every $n_1\neq n_2$, a.e. When he/she dies, the total legacy is given by:
\begin{equation}\label{legacy}
    \mathcal{J}_n(\tau):=X(\tau)+\sum_{n=1}^M\frac{p_n(\tau)}{\eta_n(\tau)}\,,
\end{equation}
where $X(\tau)$ is the wealth process of the wage earner at time $\tau\in[0,T]$.

Let $c(t)$ denote the consumption rate of the wage earner and $\pi(t)$ the fraction of the wage earner's wealth invested in the risky share at time $t$, satisfying the following integrability condition.

\begin{equation}\label{integrability}
\int_0^T[c(t)+\pi^2(t)]dt<\infty, \ \ \ \rm{a.s.}
\end{equation}

Moreover, we assume that the
shares are divisible, continuously traded and there is no transaction
costs, taxes or short-selling constraints in the trading. Then the wealth process
$X(t)$ is defined by the following (SDE):

\begin{eqnarray}\label{wealth}
   dX(t) &=& \left[X(t)(r(t)+\pi(t)\mu(t,Y(t)))-c(t)-\sum_{n=1}^Mp_n(t)\right]dt \\ \nonumber
&&  \ \ \ +\pi(t)\beta(t,Y(t))X(t)dW_1(t) +\pi(t)X(t)\sigma(t,Y(t))dW_2(t) \\ \nonumber
 && \ \ \  +\pi(t)X(t)\int_{\mathbb{R}}\gamma(t,Y(t),z)\widetilde{N}(dt,dz)\,, \ \ \
t\in(0,\tau\wedge T]\,, \\ \nonumber
   X(0) &=& x>0\,,
 \end{eqnarray}
where $\mu(t,Y(t)):=\alpha(t,Y(t))-r(t)$ is the appreciation rate and $\tau\wedge
T:=\min\{\tau, T\}$. We assume that $\mu(t,Y(t))>0$, i.e., the expected return of the risk share is higher than the interest rate.

Let $\rho(t)>0$ be deterministic process denoting the discount rate process. We define the utility functions $U_i:[0,T]\times \mathbb{R}_+\rightarrow\mathbb{R}_+\,, \ \ i=1,2,3$ as the concave, non-decreasing, continuous and differentiable functions with respect to the second variable, and the strictly decreasing continuous inverse functions $ I_i:[0,T]\times\mathbb{R}_+\rightarrow\mathbb{R}_+\,, \ \ i=1,2,3\,,$ by
\begin{equation}\label{inverse}
I_i(t,x)=\left(\frac{\partial U_i(t,x)}{\partial x}\right)^{-1}\,, \ \ i=1,2,3\,.
\end{equation}

Let $p(t):=(p_1(t),\ldots,p_M(t))$ be the vector of the insurance rates paid at the insurance companies.
The wage earner faces the problem of choosing the optimal strategy $\mathcal{A}:=\{(\pi,c,p):=(\pi(t),c(t),p(t))_{t\in[0,T]}$\} which maximizes the discounted expected utilities from the consumption during his/her lifetime $[0,\tau\wedge T]$, from the wealth if he/she is alive until the terminal time $T$ and from the legacy if he/she dies before time $T$. This problem can be defined by the following performance functional (for more details see, e.g., Pliska and Ye \cite{Pliska}, Oksendal and Sulem \cite{Oksendal}, Azevedo {et. al.} \cite{azevedo}, Guambe and Kufakunesu \cite{guambe}).

\begin{eqnarray}\nonumber
   && J(0,x,\pi,c,p) \\ \nonumber
  \quad &:=&
\sup_{(\pi,c,p)\in\mathcal{A}}\mathbb{E}\Bigl[\int_0^{\tau\wedge T}
e^{-\int_0^s\rho(u)du}U_1(s,c(s))ds  \\ \label{maximumutility} &&
+e^{-\int_0^{\tau}\rho(u)du}U_2(\tau,\mathcal{J}_n(\tau))\mathbf{1}_{\{\tau\leq T\}}
+e^{-\int_0^T\rho(u)du}U_3(X(T))\mathbf{1}_{\{\tau>T\}}\Bigl],
\end{eqnarray}
where $\mathbf{1}_A$ is a characteristic function of the set $A$.

The set of strategies $\mathcal{A}:=\{(\pi,c,p):=(\pi(t),c(t),p(t))_{t\in[0,T]}$\} is said to be admissible if, in addition to the integrability condition \eqref{integrability}, the SDE \eqref{wealth} has a unique strong solution such that $X(t)\geq0$, $\mathbb{P}$-a.s. and

\begin{eqnarray*}
   &&
\mathbb{E}\Bigl[\int_0^{\tau\wedge T}
e^{-\int_0^s\rho(u)du}U_1(s,c(s))ds +e^{-\int_0^{\tau}\rho(u)du}U_2(\tau,\mathcal{J}(\tau))\mathbf{1}_{\{\tau\leq T\}}  \\ \nonumber && \ \ \ \ \ \ \ \
+e^{-\int_0^T\rho(u)du}U_3(X(T))\mathbf{1}_{\{\tau>T\}}\Bigl]<\infty\,.
\end{eqnarray*}

Note that from the conditional survival probability of the wage earner \eqref{survival} and the conditional survival probability density of death of the wage earner \eqref{death}, we can write the dynamic version of the functional \eqref{maximumutility} by:
\begin{eqnarray}\nonumber
  J(t,x,\pi,c,p) &=&  \mathbb{E}_{t,x}\Bigl[\int_t^T
e^{-\int_t^s(\rho(u)+\lambda(u))du}[U_1(s,c(s)) +\lambda(s)U_2(s,\mathcal{J}(s))]ds \\ \label{functional0}
   && \ \ \ \ \ \  +e^{-\int_t^T(\rho(u)+\lambda(u))du}U_3(X(T))\Bigl].
\end{eqnarray}
Thus, the problem of the wage earner is to maximize the above dynamic performance functional under the admissible strategy $\mathcal{A}$. Therefore, the value function  $V(t,x,y)$ can be restated in the following form:
\begin{equation}\label{valuefunction}
V(t,x,y)=\sup_{(\pi,c,p)\in\mathcal{A}}J(t,x,\pi,c,p)\,.
\end{equation}

 Applying the results in the previous section to solve the above problem, we define the Hamiltonian $\mathcal{H}:[0,T]\times \mathbb{R}\times\mathbb{R}\times\mathbb{R}\times(0,1)\times\mathbb{R}^M\times\mathbb{R}\times\mathbb{R}\times\mathbb{R} \times\mathbb{R}\times\mathbb{R}\rightarrow\mathbb{R}$ by:

\begin{eqnarray}\label{hamiltonian}
  && \mathcal{H}(t,X(t),Y(t),c(t),\pi(t),p(t),A_1(t),A_2(t),B_1(t),B_2(t), B_3(t),D_1(t))  \\ \nonumber
  &=& e^{-\int_0^t(\rho(s)+\lambda(s))ds}[U_1(t,c(t)) +\lambda(t)U_2(t,\mathcal{J}(t))] \\ \nonumber
  && +\left[X(t)(r(t)+\pi(t)\mu(t,Y(t)))-c(t)-\sum_{n=1}^Mp_n(t)\right]A_1(t) +g(Y(t))A_2(t)  \\ \nonumber
   && +\pi(t)X(t)(\beta(t,Y(t))B_1(t)+\sigma(t,Y(t))B_2(t))+B_3(t) \\ \nonumber
   && +\pi(t)X(t)\int_{\mathbb{R}}\gamma(t,Y(t),z)D_1(t,z)\nu(dz)\,.
\end{eqnarray}

The adjoint equations corresponding to the admissible strategy $(\pi,c,p)$ are given by the following BSDEs

\begin{eqnarray}\nonumber
  dA_1(t) &=& -\frac{\partial \mathcal{H}}{\partial x}(t,X(t),Y(t),c(t),\pi(t),p(t),A_1(t),A_2(t),B_1(t),B_2(t), B_3(t),D_1(t))dt \\ \label{adjoint1}
   && +B_1(t)dW_1(t) +B_2(t)dW_2(t)+  \int_{\mathbb{R}}D_1(t,z)\tilde{N}(dt,dz)\,; \\ \nonumber
   A_1(T) &=& e^{-\int_0^T(\rho(s)+\lambda(s))ds}U_3'(X(T))\,,
\end{eqnarray}
where $U':=U_x$ and
\begin{eqnarray}\nonumber
  dA_2(t) &=& -\frac{\partial \mathcal{H}}{\partial y}(t,X(t),Y(t),c(t),\pi(t),p(t),A_1(t),A_2(t),B_1(t),B_2(t), B_3(t),D_1(t))dt \\ \label{adjoint2}
  && +B_1(t)dW_3(t) +B_4(t)dW_2(t)+  \int_{\mathbb{R}}D_2(t,z)\tilde{N}(dt,dz)\,; \\ \nonumber
  A_2(T) &=& 0\,.
\end{eqnarray}

To solve our optimization problem, we consider the power utility functions of the constant relative risk aversion type defined as follows
 $U_i(t,x)=U_i(x)=\kappa_i\frac{x^{\delta}}{\delta}\,, \ \ \ i=1,2,3$, where $\delta\in(-\infty,1)\setminus\{0\}$ and $\kappa_i>0$ are constants. Thus, the inverse function \eqref{inverse} is given by $I_i(t,x)=I_i(x)=\left(\frac{x}{\kappa_i}\right)^{-\frac{1}{1-\delta}}$\,.\\

The following theorem gives the characterization of the optimal strategy.
\begin{theorem}\label{theoremmulti}
Suppose that the assumptions $\bf{(A1)}-\bf{(A2)}$ and the integrability condition \eqref{integrability} hold. Then the optimal strategy $(c^*, p^*,\pi^*)\in\mathcal{A}$ for the problem \eqref{valuefunction} is given by:
\begin{itemize}
  \item[$(i)$] the optimal consumption process is given by
\begin{eqnarray}\nonumber
c^*(t,x,y) &=& I_1\left(t,\frac{A_1^*(t)}{\kappa_1}(t)e^{\int_0^t(\rho(s)+\lambda(s))ds}\right) \\ \label{consumption}
&=& \left(\frac{A_1^*(t)}{\kappa_1}\right)^{\frac{1}{\delta-1}}e^{\frac{1}{\delta-1}\int_0^t(\rho(s)+\lambda(s))ds}\,;
\end{eqnarray}

  \item[$(ii)$] for each $n\in\{1,2,\ldots,M\}$, the optimal premium insurance $p_n(t,x,y)$ is given by
  \begin{eqnarray}\nonumber
  p_n^*(t,x,y) &=& \left\{
                 \begin{array}{ll}
                   \max\left\{0,\,\left[I_2\left(t,\frac{\eta_n(t)}{\kappa_2\lambda(t)}A_1^*(t)e^{\int_0^t(\rho(s)+\lambda(s))ds}\right)-x \right] \right\}, & \hbox{if} \ \  n=n^*(t) \\
                   0, & \hbox{otherwise}\,,
                 \end{array}
               \right. \\ \label{insurance}
   &=& \left\{
                 \begin{array}{ll}
                   \max\left\{0,\,\eta_n(t)\left[\left(\frac{\eta_n(t)A_1^*(t)}{\kappa_2\lambda(t)}\right)^{\frac{1}{\delta-1}} e^{\frac{1}{\delta-1}\int_0^t(\rho(s)+\lambda(s))ds}-x \right] \right\}, & \hbox{if} \ \  n=n^*(t) \\
                   0, & \hbox{otherwise}\,,
                 \end{array}
               \right.
  \end{eqnarray}
where $n^*(t)=\arg\min_{n\in\{1,2,\ldots,M\}}\{\eta_n(t)\}$
\item[$(iii)$] and, the optimal allocation $\pi^*(t,x,y)\in(0,1)$ is the solution of the following equation
  \begin{eqnarray*}
\beta(t,y)h_y(t,y)-\Bigl\{\mu(t,y) -(1-\delta)\left(\beta^2(t,y)+\sigma^2(t,y)\right)\pi && \\
- \int_{\mathbb{R}}\left[1-(1+\pi\gamma(t,y,z))^{\delta-1}\right]\gamma(t,y,z)\nu(dz)\Bigl\}h(t,y) &=& 0\,,
\end{eqnarray*}
where $h\in \mathcal{C}^{1,2}([0,T]\times\mathbb{R})$.
  \end{itemize}

\end{theorem}

\begin{proof}
From the Hamiltonian function \eqref{hamiltonian} and the definition of the utility functions $U_1,\,U_2$, we can deduce the following conditions:
\begin{eqnarray*}
\mathcal{H}_{cc} &=& e^{-\int_0^t(\rho(s)+\lambda(s))ds}\frac{\partial^2U_1}{\partial c^2}(t,c)<0\,, \\
\mathcal{H}_{p_{n_1}p_{n_2}} &=& e^{-\int_0^t(\rho(s)+\lambda(s))ds}\frac{\lambda(t)}{\eta_{n_1}\eta_{n_2}}\frac{\partial^2U_2}{\partial x^2}\left(t,x+\sum_{n=1}^M\frac{p_n}{\eta_n(t)}\right)<0\,.
\end{eqnarray*}
Thus, it is sufficient to obtain the optimal consumption and insurance $(c^*,p^*)$ by applying the first order conditions of optimality. Then from \eqref{hamiltonian} we have the following:

\begin{itemize}
  \item[$(i)$] The optimal consumption $c^*(t,x,y)$ is obtained from the following
\begin{equation*}
    -A_1(t)+e^{-\int_0^t(\rho(s)+\lambda(s))ds}\frac{\partial U_1}{\partial c}(t,c)=0\,.
\end{equation*}
From \eqref{inverse}, the optimal consumption can explicitly be obtained by
\begin{eqnarray*}
c^*(t,x,y) &=& I_1\left(t,\frac{A_1^*(t)}{\kappa_1}e^{\int_0^t(\rho(s)+\lambda(s))ds}\right) \\
&=& \left(\frac{A_1^*(t)}{\kappa_1}\right)^{\frac{1}{\delta-1}}e^{\frac{1}{\delta-1}\int_0^t(\rho(s)+\lambda(s))ds}\,;
\end{eqnarray*}

  \item[$(ii)$] The optimal premium insurance $p^*_n(t,x,y)$ is obtained using the Kuhn-Tucker conditions of optimality. As in Mousa {\it et al.} \cite{mousa}, we are looking for the solutions $(p_1(t,x,y);\ldots;p_M(t,x,y);\xi_1(t,x,y);\ldots;\xi_M(t,x,y))$ in the following system:
\begin{equation}\label{kuhn-tucker}
    \left\{
      \begin{array}{ll}
        -A_1(t)+\frac{\lambda(t)}{\eta_n(t)}e^{-\int_t^0(\rho(s)+\lambda(s))ds} \frac{\partial U_2}{\partial x}\left(t,x+\sum_{n=1}^M\frac{p_n}{\eta_n(t)}\right)=-\xi_n(t,x,y) & \hbox{} \\
        p_n(t,x,y)\geq0\,; \ \ \ \xi_n(t,x,y)\geq0\,; & \hbox{} \\
        p_n(t,x,y)\xi_n(t,x,y)=0, \ \ \ \forall n=1,2,\ldots M\,. & \hbox{}
      \end{array}
    \right.
\end{equation}
First, suppose that $n_1\neq n_2$. If we have $\xi_{n_1}(t,x,y)=\xi_{n_2}(t,x,y)$, for some $(t,x,y)\in[0,T]\times\mathbb{R}\times\mathbb{R}$, one must have $\eta_{n_1}(t)=\eta_{n_2}(t)$. Thus, from the assumption that all the insurance companies offer distinct contracts, we obtain that for every $n_1,n_2\in\{1,2,\ldots,M\}$, such that $n_1\neq n_2$, then $\xi_{n_1}(t,x,y)\neq\xi_{n_2}(t,x,y)$; $(t,x,y)\in[0,T]\times\mathbb{R}\times\mathbb{R}$, a.e. Therefore, there is at most one $n\in\{1,2,\ldots,M\}$ such that $p_n(t,x,y)\neq0$.

Then from the first equation in the system \eqref{kuhn-tucker},
$$
\eta_{n_1}(A_1(t)-\xi_{n_1}(t,x,y))=\eta_{n_2}(A_1(t)-\xi_{n_2}(t,x,y))\,.
$$
Hence, we can conclude that if $\xi_{n_1}(t,x,y)>\xi_{n_2}(t,x,y)$, then $\eta_{n_1}(t)>\eta_{n_2}(t)$. Moreover, if $\xi_{n_1}(t,x,y)=0$ for some $t\in[0,T]$, $\eta_{n_1}(t)<\eta_{n_2}(t)$, $\forall n_2\in\{1,2,\ldots,M\}$ such that $n_1\neq n_2$. From this point, let $n^*(t)=\arg\min_{n\in\{1,2,\ldots,M\}}\{\eta_n(t)\}$, then either $p_n(t,x,y)=0$ or $p_{n^*}(t,x,y)>0$ is the solution to the equation
$$
-A_1(t)+\frac{\lambda(t)}{\eta_{n^*}(t)}e^{-\int_t^0(\rho(s)+\lambda(s))ds}\frac{\partial U_2}{\partial x}\left(t,x+\frac{p_{n^*}}{\eta_{n^*}(t)}\right)=0,
$$
which gives the required solution
\begin{eqnarray*}
  p_n^*(t,x,y) &=& \left\{
                 \begin{array}{ll}
                   \max\left\{0,\,\left[I_2\left(t,\frac{\eta_n(t)}{\kappa_2\lambda(t)}A_1^*(t)e^{\int_0^t(\rho(s)+\lambda(s))ds}\right)-x \right] \right\}, & \hbox{if} \ \  n=n^*(t) \\
                   0, & \hbox{otherwise}\,,
                 \end{array}
               \right. \\
   &=& \left\{
                 \begin{array}{ll}
                   \max\left\{0,\,\eta_n(t)\left[\left(\frac{\eta_n(t)A_1^*(t)}{\kappa_2\lambda(t)}\right)^{\frac{1}{\delta-1}} e^{\frac{1}{\delta-1}\int_0^t(\rho(s)+\lambda(s))ds}-x \right] \right\}, & \hbox{if} \ \  n=n^*(t) \\
                   0, & \hbox{otherwise}\,;
                 \end{array}
               \right.
  \end{eqnarray*}

\item[$(iii)$] Since the expression involving $\pi$ in the Hamiltonian $\mathcal{H}$ \eqref{hamiltonian} is linear, for the maximum investment $\pi^*$, we have the following relation:
\begin{equation}\label{investmentr}
    \mu(t,y)A_1^*(t)+\beta(t,y)B_1^*(t)+\sigma(t,y)B_2^*(t)+\int_{\mathbb{R}}\gamma(t,y,z)D_1^*(t,z)\nu(dz)=0\,.
\end{equation}

  \end{itemize}
To obtain the optimal portfolio, we first solve the adjoint BSDE equations \eqref{adjoint1} and \eqref{adjoint2}. From the terminal condition of the adjoint equation \eqref{adjoint1}, we try the solution of the first adjoint equation $A_1^*(t)$ of the form
\begin{equation}\label{solutiona1}
    A_1^*(t)=X(t)^{\delta-1}e^{-h(t,Y(t))}\,, \ \ \ h(T,Y(T))=\int_0^T(\rho(u)+\lambda(u))du\,.
\end{equation}
On the other hand, for the optimal strategy $(c^*,p_n^*,\pi^*)$, we have

\begin{equation}\label{adjointo1}
dA_1^*(t)=-\eta_{n^*}A_1^*(t)dt+B_1^*(t)dW_1(t)+B_2^*(t)dW_2(t)+\int_{\mathbb{R}}D_1^*(t,z)\tilde{N}(dt,dz)\,.
\end{equation}
Applying the It\^o's product rule in \eqref{solutiona1} and from \eqref{wealth}, \eqref{consumption} and \eqref{insurance}, we obtain

\begin{eqnarray*}
   dA_1^*(t) &=& -x(t)^{\delta-1}e^{-h(t,y)}\Bigl\{h_t(t,y) +g(y)h_y(t,y)+\frac{1}{2}h_{yy}(t,y)-\frac{1}{2}(h_y(t,y))^2 \\
  && +\frac{1}{2}(\delta-1)\pi^*(t)\beta(t,y)h_y(t,y) -\Bigl[(\delta-1)[r(t)+\mu(t,y)\pi^*(t)+\eta_{n^*}(t)] \\
   && +\frac{1}{2}(\delta-1)(\delta-2)(\pi^*(t))^2 (\beta^2(t,y)+\sigma^2(t,y)) \\
   && +\int_{\mathbb{R}}\bigl[(1+\pi^*(t)\gamma(t,y,z))^{\delta-1}-1-(\delta-1)\pi^*(t)\gamma(t,y,z)\bigl]\nu(dz)\Bigl] \\
   &&  +(1-\delta)e^{\frac{1}{1-\delta}h(t,y)}e^{\int_0^t(\rho(s)+\lambda(s))ds}\Bigl[1+\eta_{n^*}(t) \left(\frac{\eta_{n^*}(t)}{\kappa_2\lambda(t)}\right)^{\frac{1}{\delta-1}}\Bigl]\Bigl\}dt \\
   && +((\delta-1)\pi^*(t)\beta(t,y)-h_y(t,y))x(t)^{\delta-1}e^{-h(t,y)}dW_1(t) \\
   && +(\delta-1)\pi^*(t)x(t)^{\delta-1}\sigma(t,y)e^{-h(t,y)}dW_2(t) \\
   && +x(t)^{\delta-1}e^{-h(t,y)}\int_{\mathbb{R}}\bigl[(1+\pi^*(t)\gamma(t,y,z))^{\delta-1}-1\bigl]\tilde{N}(dt,dz)\,.
\end{eqnarray*}
Comparing with the adjoint equation \eqref{adjointo1}, we get:

\begin{eqnarray}\label{optimalb1}
  B^*_1(t) &=& ((\delta-1)\pi^*(t)\beta(t,y)-h_y(t,y))x(t)^{\delta-1}e^{-h(t,y)}\,; \\ \label{optimalb2}
  B_2^*(t) &=& (\delta-1)\pi^*(t)\sigma(t,y)x(t)^{\delta-1}e^{-h(t,y)}\,; \\ \label{optimald1}
  D_1^*(t) &=& x(t)^{\delta-1}e^{-h(t,y)}\bigl[(1+\pi^*(t)\gamma(t,y,z))^{\delta-1}-1\bigl]\,.
\end{eqnarray}
Furthermore, $h$ is a solution of the following backward partial differential equation (PDE)
\begin{eqnarray}\nonumber
    h_t(t,y) +(g(y)+\frac{1}{2}(\delta-1)\pi^*(t)\beta(t,y))h_y(t,y) +\frac{1}{2}h_{yy}(t,y) -\frac{1}{2}(h_y(t,y))^2 && \\ \label{ode}
     +K(t) + (1-\delta)e^{\frac{1}{1-\delta}h(t,y)}e^{\int_0^t(\rho(s)+\lambda(s))ds}\Bigl[1+\eta_{n^*}(t) \left(\frac{\eta_{n^*}(t)}{\kappa_2\lambda(t)}\right)^{\frac{1}{\delta-1}}\Bigl] &=& 0\,,
\end{eqnarray}
with a terminal condition $h(T,Y(T))=e^{-\int_0^T(\rho(u)+\lambda(u))du}$, where
\begin{eqnarray*}
  K(t) &=& -(\delta-1)[r(t)+\mu(t,y)\pi^*(t)+\delta\eta_{n^*}(t) \\
  && +\frac{1}{2}(\delta-1)(\delta-2)(\pi^*(t))^2 (\beta^2(t,y)+\sigma^2(t,y)) \\
   && +\int_{\mathbb{R}}\bigl[(1+\pi^*(t)\gamma(t,y,z))^{\delta-1}-1-(\delta-1)\pi^*(t)\gamma(t,y,z)\bigl]\nu(dz)\,.
\end{eqnarray*}
The solution of the above equation can be approximated by a fixed point algorithm. To that end, we define a Feynman-Kac operator $\Phi$, acting on functions $h$ as follows
\begin{eqnarray*}
(\Phi h)(t,Y(t)) &=& \mathbb{E}\Bigl[e^{Q(t,Y(t))} +\int_t^Te^{Q(s,Y(s))}\Bigl\{K(s) \\
&&  + (1-\delta)e^{\frac{1}{1-\delta}h(s,Y(s))}e^{\int_0^s(\rho(u)+\lambda(u))du}\Bigl[1+\eta_{n^*}(s) \left(\frac{\eta_{n^*}(s)}{\kappa_2\lambda(s)}\right)^{\frac{1}{\delta-1}}\Bigl]\Bigl\}ds\Bigl]\,,
\end{eqnarray*}
where $Q(t,Y(t))=g(Y(t))+\frac{1}{2}(\delta-1)\pi^*(t)\beta(t,Y(t))$.

To solve \eqref{ode}, one need to find a fixed point solution, for the following fixed point equation
$$
(\Phi h)(t,y)=h(t,y).
$$
Then, under the Assumptions $({\bf A1})$ and $({\bf A2})$, there exists a unique solution $h\in \mathcal{C}^{1,2}([0,T]\times\mathbb{R})$ of the PDE \eqref{ode}. (See, Berdjane and Pergamenshchikov \cite{Berdjane}, Theorem 3.1.)\\

Now, substituting \eqref{solutiona1}, \eqref{optimalb1}, \eqref{optimalb2}, \eqref{optimald1} into \eqref{investmentr}, we obtain
\begin{eqnarray*}
&& \beta(t,y)h_y(t,y) -\Bigl\{\mu(t,y)+(\delta-1)\pi^*(t)(\beta^2(t,y)+\sigma^2(t,y)) \\
&& +\int_{\mathbb{R}}\gamma(t,y,z)\bigl[(1+\pi^*(t)\gamma(t,y,z))^{\delta-1}-1\bigl]\nu(dz)\Bigl\} \,=\,0\,.
\end{eqnarray*}
Note that the second derivative in $\pi$, for the optimal $A_1^*,B_1^*,B_2^*$ and $D_1^*$ is negative, i.e.,
$$
-(1-\delta)\Bigl[\beta^2(t,y)+\sigma^2(t,y) +\int_{\mathbb{R}}(1+\pi(t)\gamma(t,y,z))^{\delta-2}\gamma^2(t,y,z)\nu(dz)\Bigl]<0.
$$
Hence, there exists an optimal $\pi^*(t)\in(0,1)$.

\end{proof}

For the second adjoint equation, note that from \eqref{investmentr}, we obtain the following relation
\begin{equation}\label{investmentrd}
    \frac{\partial\mu}{\partial y}(t,y)A_1^*(t) +\frac{\partial\beta}{\partial y}(t,y)B_1^*(t) +\frac{\partial\sigma(t,y)}{\partial y}B_2^*(t)+\int_{\mathbb{R}}\frac{\partial\gamma}{\partial y}(t,y,z)D_1^*(t,z)\nu(dz)=0\,.
\end{equation}

Then, for optimal strategy, the second adjoint equation \eqref{adjoint2}, can be written as
\begin{equation}\label{adjointo2}
dA_2^*(t)=-g'(y)A_2^*(t)dt+B_3^*(t)dW_1(t) +B_4^*(t)dW_1(t) +\int_{\mathbb{R}}D_2^*(t,z)\tilde{N}(dt,dz)\,,
\end{equation}
which is a linear BSDE with jumps. Since the terminal condition is $A^*(T)=0$, by applying the techniques for solving linear BSDE with jumps (Delong \cite{Delong}, Propositions 3.3.1 and 3.4.1), we obtain $A_2^*(t)=B_3^*(t)=B_4^*(t)=D_2^*(t,z)=0$.

The corresponding wealth process equation \eqref{wealth} for the optimal solutions becomes

\begin{equation*}
  dX^*(t) = X(t)\Bigl[G(t)dt+\pi^*(t)[\beta(t,y)dW_1(t)+\sigma(t,y)dW_2(t)] +\pi^*(t)\int_{\mathbb{R}}\gamma(t,y,z)\tilde{N}(dt,dz)\Bigl],
\end{equation*}
where
$$
G(t)=r(t)+\pi^*(t)\mu(t,y)+\eta_{n^*}(t)-h(t)^{\frac{1}{\delta-1}}\left[\kappa_1^{-\frac{1}{\delta-1}} +\left(\frac{\eta_{n^*}(t)}{\kappa_2\lambda(t)}\right)^{\frac{1}{\delta-1}} e^{\frac{1}{\delta-1}\int_0^t(\rho(s)+\lambda(s))ds}\right]\,,
$$
which gives the following solution

\begin{eqnarray*}
  X(t) &=& x\exp\Bigl\{\int_0^t[G(s)-\frac{1}{2}(\pi^*(s))^2(\beta^2(s,y)+\sigma^2(s,y))]ds \\
  &&  +\int_0^t\int_{\mathbb{R}}[\ln(1+\pi^*(s)\gamma(s,y,z))-\pi^*(s)\gamma(s,y,z)]\nu(dz)ds \\
   && +\int_0^t\pi^*(s)[\beta(s,y)dW_1(s)+\sigma(s,y)dW_2(s)] \\
   &&  +\int_0^t\int_{\mathbb{R}}\ln(1+\pi^*(s)\gamma(s,y,z))\tilde{N}(ds,dz)\Bigl\}\,.
\end{eqnarray*}

Finally, the value function of the problem \eqref{valuefunction} can be characterized as the solution of the following BSDE

\begin{eqnarray*}
  dV(t,x,y) &=& -\mathcal{H}(t,x^*,y,c^*,\pi^*,p^*,A_1^*,A_2^*,B_1^*,B_2^*,D_1^*)dt+B_1^*(t)dW_1(t)+B_2^*(t)dW_2(t) \\
   && +\int_{\mathbb{R}}D_1^*(t,z)\tilde{N}(dt,dz)\,; \\
  V(T,x,y) &=& \kappa_3e^{\int_0^T[\rho(t)+\lambda(t)]dt}\frac{X(T)^{\delta}}{\delta}\,.
\end{eqnarray*}

\begin{example}
The following example specifies the results in Theorem \ref{theoremmulti} to a well known stochastic volatility model of Ornstein-Uhlenbeck type and an explicit portfolio strategy is derived. Let $N$ be the Poisson process, with intensity $\nu>0$. We consider the following model dynamics
\begin{eqnarray*}
  B(t) &=& 1\,; \\
  dS(t) &=& S(t)[(\alpha_0+\alpha_1Y(t))dt+\gamma Y(t) d\tilde{N}(t)]\,; \\
  dY(t) &=& -bY(t)dt+dW(t)\,,
\end{eqnarray*}
where $\alpha_0,\alpha_1,\gamma\in\mathbb{R}$ and $b>0$.
Suppose that we have a constant mortality rate $\lambda>0$, constants insurance premium rate $\eta_n>0, \ \  n=1,2,\ldots,M$, discount rate $\rho>0$ and $\kappa_1=\kappa_1=\kappa_3=1$. Then the Hamiltonian is given by
\begin{eqnarray*}
  && \mathcal{H}(t,X(t),Y(t),A_1(t),A_2(t),B(t),D_1(t),D_2(t)) \\
   &=& \frac{1}{\delta}e^{-(\rho+\lambda)t}\Bigl[(c(t))^{\delta}+\lambda(X+\sum_{n=1}^M\frac{p_n(t)}{\eta_n})^{\delta}\Bigl] \\
   && +[X(t)\pi(t)(\alpha_0+\alpha_1Y(t))-c(t)-\sum_{n=1}^M\frac{p_n(t)}{\eta_n}]A_1(t) \\
   &&  -bY(t)A_2(t)+B(t)+\pi(t)X(t)\gamma Y(t) D_1(t)\nu+D_2(t)\,.
\end{eqnarray*}
Then, following Theorem \ref{theoremmulti}, we can easily see that the optimal portfolio is given by
$$
\pi^*(t)=\frac{1}{\delta}\Bigl[\Bigl(\frac{\gamma\nu y-\alpha_0-\alpha_1y}{\gamma\nu y}\Bigl)^{\frac{1}{\delta-1}}-1\Bigl]\,,
$$
where $y$ is given by $Y(t)=e^{-bt}y_0+\int_0^te^{-b(t-s)}dW(s)$.\\
Moreover, the optimal consumption and insurance are given by
$$
c^*(t)=e^{\frac{1}{\delta-1}(\rho+\lambda)t}(A_1^*(t))^{\frac{1}{\delta-1}}\,, \ \ \ p_{n^*}^*(t)=\Bigl[ \Bigl(\frac{\eta_{n^*}}{\lambda}A_1^*(t)\Bigl)^{\frac{1}{\delta-1}} e^{\frac{1}{\delta-1}(\rho+\lambda)t}-x\Bigl]\,,
$$
where $A_1^*(t)$ is part of a solution of the following linear BSDE
$$
dA_1^*(t)=-\eta_{n^*}A^*_1(t)dt+B^*(t)dW(t)+D^*(t)d\tilde{N}(t).
$$
Hence, $A_1^*(t)=e^{-\rho T}\mathbb{E}\Bigl[e^{\eta_{n^*}(T-t)}(X(T))^{\delta-1}\mid\mathcal{F}_t\Bigl]$. $B^*$ and $D^*$ can be derived by the martingale representation theorem. See Delong \cite{Delong}, Propositions 3.3.1 and 3.4.1. Thus, for this pure jump Poisson process of Ornstein-Uhlenbeck type, we have derived an explicit optimal portfolio strategy.
\end{example}

\subsection*{Acknowledgment}

We would like to express our deep gratitude to the NRF Project No: CSUR 90313, the University of Pretoria and the MCTESTP Mozambique for their support.

\section*{Appendix. Proof of the main results}

\noindent {\it Proof of Theorem \ref{sufficientmp}.}
Let $\pi\in\mathcal{A}$ be an admissible strategy and $X(t)$ the corresponding wealth process. Then, following Framstad {\it et. al.} \cite{framstard}, Theorem 2.1., we have:

\begin{eqnarray*}
 \mathcal{J}(\pi^*)-\mathcal{J}(\pi)
  &=& \mathbb{E}\Bigl[\int_0^T(f(t,X^*(t),Y^*(t),\pi^*(t))-f(t,X(t),Y(t),\pi(t)))dt \\
  && +(g(X^*(T),Y^*(T))-g(X(T),Y(T))) \Bigl] \\
  &=:& \mathcal{K}_1+\mathcal{K}_2\,.
\end{eqnarray*}
By condition $(i)$ and the integration by parts rule (Oksendal and Sulem \cite{Oksendal}, Lemma 3.6.), we have

\begin{eqnarray*}
  \mathcal{K}_2 &=& \mathbb{E}\Bigl[g(X^*(T),Y^*(T))-g(X(T),Y(T))\Bigl] \\
  &\geq& \mathbb{E}\Bigl[(X^*(T)-X(T))A^*_1(T)+(Y^*(T)-Y(T))A^*_2(T)\Bigl]  \\
   &=& \mathbb{E}\Bigl[\int_0^T(X^*(t)-X(t))dA^*_1(t)+\int_0^TA^*_1(t)(dX^*(t)-dX(t)) \\
   && +\int_0^T(Y^*(t)-Y(t))dA^*_2(t)+\int_0^TA^*_2(t)(dY^*(t)-dY(t)) \\
   && +\int_0^T[(\sigma(t,X^*(t),Y^*(t),\pi^*(t))- \sigma(t,X(t),Y(t),\pi(t)))B^*_1(t) \\
  &&  +  (\beta(t,X^*(t),Y^*(t),\pi^*(t))- \sigma(t,X(t),Y(t),\pi(t)))B^*_2(t)]dt \\
   && + \int_0^T\int_{\mathbb{R}}(\gamma(t,X^*(t),Y^*(t),\pi^*(t),z)-\gamma(t,X^*(t),Y^*(t),\pi^*(t),z))D_1^*(t,z)\nu(dz)dt \\
   &&  +\int_0^T(\phi(Y^*(t))-\phi(Y(t)))B_3^*(t)dt\Bigl] \\
   &=& \mathbb{E}\Bigl[-\int_0^T(X^*(t)-X(t))\frac{\partial \mathcal{H}^*}{\partial x}(t)dt -\int_0^T(Y^*(t)-Y(t))\frac{\partial \mathcal{H}^*}{\partial y}(t)dt \\
   && + \int_0^T(A_1^*(t)b(t,X^*(t),Y^*(t),\pi^*(t))-b(t,X(t),Y(t),\pi(t)))dt \\
   && + \int_0^T(\varphi(Y^*(t))-\varphi(Y(t)))A_2^*(t)dt +\int_0^T(\phi(Y^*(t))-\phi(Y(t)))B_3^*(t)dt \\
   && +\int_0^T[(\sigma(t,X^*(t),Y^*(t),\pi^*(t))- \sigma(t,X(t),Y(t),\pi(t)))B^*_1(t) \\
  &&  +  (\beta(t,X^*(t),Y^*(t),\pi^*(t))- \sigma(t,X(t),Y(t),\pi(t)))B^*_2(t)]dt \\
   && + \int_0^T\int_{\mathbb{R}}(\gamma(t,X^*(t),Y^*(t),\pi^*(t),z) -\gamma(t,X^*(t),Y^*(t),\pi^*(t),z))D_1^*(t,z)\nu(dz)dt\Bigl]\,,
\end{eqnarray*}
where we have used the notation $$\mathcal{H}^*(t)=\mathcal{H}(t,X^*(t),Y^*(t),\pi^*(t),A_1^*(t),A_2^*(t),B_1^*(t),B_2^*(t),B_3^*(t),D_1^*(t,\cdot))\,.$$

On the other hand, by definition of $\mathcal{H}$ in \eqref{hamiltoniangeneral}, we see that

\begin{eqnarray*}
 \mathcal{K}_1 &=& \mathbb{E}\Bigl[\int_0^T(f(t,X^*(t),Y^*(t),\pi^*(t))-f(t,X(t),Y(t),\pi(t)))dt\Bigl] \\
   &=& \mathbb{E}\Bigl[\int_0^T [\mathcal{H}(t,X^*(t),Y^*(t),\pi^*(t),A_1^*(t),A_2^*(t),B_1^*(t),B_2^*(t),B_3^*(t),D_1^*(t,\cdot))  \\
   && - \mathcal{H}(t,X^*(t),Y^*(t),\pi^*(t),A_1^*(t),A_2^*(t),B_1^*(t),B_2^*(t),B_3^*(t),D_1^*(t,\cdot))]dt  \\
   && -\int_0^TA_1^*(t)(A_1^*(t)b(t,X^*(t),Y^*(t),\pi^*(t))-b(t,X(t),Y(t),\pi(t)))dt \\
   && - \int_0^T(\varphi(Y^*(t))-\varphi(Y(t)))A_2^*(t)dt +\int_0^T(\phi(Y^*(t))-\phi(Y(t)))B_3^*(t)dt \\
   && -\int_0^T[(\sigma(t,X^*(t),Y^*(t),\pi^*(t))- \sigma(t,X(t),Y(t),\pi(t)))B^*_1(t) \\
  &&  -  (\beta(t,X^*(t),Y^*(t),\pi^*(t))- \sigma(t,X(t),Y(t),\pi(t)))B^*_2(t)]dt \\
   && - \int_0^T\int_{\mathbb{R}}(\gamma(t,X^*(t),Y^*(t),\pi^*(t),z) -\gamma(t,X^*(t),Y^*(t),\pi^*(t),z))D_1^*(t,z)\nu(dz)dt\Bigl]\,.
\end{eqnarray*}
Then, summing the above two expressions, we obtain
\begin{eqnarray*}
 && \mathcal{K}_1+\mathcal{K}_2 \\
  &=&  \mathbb{E}\Bigl[\int_0^T [\mathcal{H}(t,X^*(t),Y^*(t),\pi^*(t),A_1^*(t),A_2^*(t),B_1^*(t),B_2^*(t),B_3^*(t),D_1^*(t,\cdot))  \\
   && - \mathcal{H}(t,X(t),Y(t),\pi(t),A_1^*(t),A_2^*(t),B_1^*(t),B_2^*(t),B_3^*(t),D_1^*(t,\cdot))]dt \\
   &&  -\int_0^T(X^*(t)-X(t))\frac{\partial \mathcal{H}^*}{\partial x}(t)dt -\int_0^T(Y^*(t)-Y(t))\frac{\partial \mathcal{H}^*}{\partial y}(t)dt\,.
\end{eqnarray*}

By the concavity of $\mathcal{H}$, i.e., conditions $(i)$ and $(ii)$, we have
\begin{eqnarray*}
   && \mathbb{E}\Bigl[\int_0^T [\mathcal{H}(t,X^*(t),Y^*(t),\pi^*(t),A_1^*(t),A_2^*(t),B_1^*(t),B_2^*(t),B_3^*(t),D_1^*(t,\cdot))  \\
   && - \mathcal{H}(t,X(t),Y(t),\pi(t),A_1^*(t),A_2^*(t),B_1^*(t),B_2^*(t),B_3^*(t),D_1^*(t,\cdot))]dt\Bigl]  \\
   &\geq&  \mathbb{E}\Bigl[\int_0^T(X^*(t)-X(t))\frac{\partial \mathcal{H}^*}{\partial x}(t)dt +\int_0^T(Y^*(t)-Y(t))\frac{\partial \mathcal{H}^*}{\partial y}(t)dt \\
   && +\int_0^T(\pi^*(t)-\pi(t))\frac{\partial \mathcal{H}^*}{\partial \pi}(t)dt\Bigl]\,.
\end{eqnarray*}

Then, by the maximality of the strategy $\pi^*\in\mathcal{A}$ and the concavity of the Hamiltonian $\mathcal{H}$,
\begin{eqnarray*}
   && \mathbb{E}\Bigl[\int_0^T [\mathcal{H}(t,X^*(t),Y^*(t),\pi^*(t),A_1^*(t),A_2^*(t),B_1^*(t),B_2^*(t),B_3^*(t),D_1^*(t,\cdot))  \\
   && - \mathcal{H}(t,X(t),Y(t),\pi(t),A_1^*(t),A_2^*(t),B_1^*(t),B_2^*(t),B_3^*(t),D_1^*(t,\cdot))]dt\Bigl]  \\
   &\geq&  \mathbb{E}\Bigl[\int_0^T(X^*(t)-X(t))\frac{\partial \mathcal{H}^*}{\partial x}(t)dt +\int_0^T(Y^*(t)-Y(t))\frac{\partial \mathcal{H}^*}{\partial y}(t)dt\Bigl]\,.
\end{eqnarray*}
Hence  $\mathcal{J}(\pi^*)-\mathcal{J}(\pi)=\mathcal{K}_1+\mathcal{K}_2\geq0$. Therefore, $\mathcal{J}(\pi^*)\geq \mathcal{J}(\pi)$, that is, the strategy $\pi^*\in\mathcal{A}$ is optimal.
\begin{flushright}
$\square$
\end{flushright}

\noindent {\it Proof of Theorem \ref{necessarymp}.}
From \eqref{functionalgeneral}, we have that
\begin{eqnarray*}
  \frac{d}{d\ell}\mathcal{J}(\pi+\ell\zeta)\Bigl|_{\ell=0} &=& \mathbb{E}\Bigl[\int_0^T\Bigl(\frac{\partial f}{\partial x}(t)x_1(t) +\frac{\partial f}{\partial y}(t)y_1(t) +\frac{\partial f}{\partial \pi}(t)\zeta(t)\Bigl)dt  \\
   &&  +\frac{\partial g}{\partial x}(X(T),Y(T))x_1(T) +\frac{\partial g}{\partial y}(X(T),Y(T))y_1(T)\Bigl]_{\ell=0}\,.
\end{eqnarray*}
Let $$ I(t):=\mathbb{E}\Bigl[\frac{\partial g}{\partial x}(X(T),Y(T))x_1(T) +\frac{\partial g}{\partial y}(X(T),Y(T))y_1(T)\Bigl]\,.$$
By It\^os formula and the dynamics of $x_1$ and $y_1$, we get

\begin{eqnarray}\nonumber
  I(t) &=& \mathbb{E}\Bigl[A_1(T)x_1(T)+A_2(T)y_1(T)\Bigl] \\ \label{itonecessary}
   &=& \mathbb{E}\Bigl[\int_0^Tx_1(t)\Bigl(A_1(t)\frac{\partial b}{\partial x}(t) +B_1(t)\frac{\partial \sigma}{\partial x}(t) +B_2(t)\frac{\partial \beta}{\partial x}(t) \\ \nonumber
   && +\int_{\mathbb{R}}\frac{\partial \gamma}{\partial x}(t,z)D_1(t,z)\nu(dz) -\frac{\partial \mathcal{H}}{\partial x}(t)\Bigl)dt \\ \nonumber
   && +\int_0^Ty_1(t)\Bigl(A_1(t)\frac{\partial b}{\partial y}(t) +A_2(t)\varphi'(Y(t)) +B_1(t)\frac{\partial \sigma}{\partial y}(t) +B_2(t)\frac{\partial \beta}{\partial y}(t) \\ \nonumber
   && +B_4(t)\phi'(Y(t)) +\int_{\mathbb{R}}\frac{\partial \gamma}{\partial y}(t,z)D_1(t,z)\nu(dz) -\frac{\partial \mathcal{H}}{\partial y}(t)\Bigl)dt \\ \nonumber
   && +\int_0^T\zeta(t)\Bigl(A_1(t)\frac{\partial b}{\partial \pi}(t) +B_1(t)\frac{\partial \sigma}{\partial \pi}(t) +B_2(t)\frac{\partial \beta}{\partial \pi}(t) \\ \nonumber
   && +\int_{\mathbb{R}}\frac{\partial \gamma}{\partial \pi}(t,z)D_1(t,z)\nu(dz)(t)\Bigl)dt\Bigl]\,.
\end{eqnarray}
On the other hand, by definition of the Hamiltonian \eqref{hamiltoniangeneral}, we have
\begin{eqnarray*}
  \nabla_{x,y,\pi}\mathcal{H}(t) &=& \frac{\partial \mathcal{H}}{\partial x}(t)x_1(t) +\frac{\partial \mathcal{H}}{\partial y}(t)y_1(t) +\frac{\partial \mathcal{H}}{\partial \pi}(t)\zeta(t) \\
   &=& x_1(t)\Bigl[\frac{\partial f}{\partial x}(t) +A_1(t)\frac{\partial b}{\partial x}(t) +B_1(t)\frac{\partial \sigma}{\partial x}(t) +B_2(t)\frac{\partial \beta}{\partial x}(t) \\
   && +\int_{\mathbb{R}}\frac{\partial \gamma}{\partial x}(t,z)D_1(t,z)\nu(dz)\Bigl] \\
   &&  +y_1(t)\Bigl[\frac{\partial f}{\partial y}(t) +A_1(t)\frac{\partial b}{\partial y}(t) +A_2(t)\varphi'(Y(t)) +B_1(t)\frac{\partial \sigma}{\partial y}(t) +B_2(t)\frac{\partial \beta}{\partial y}(t) \\
   && +B_3(t)\phi'(Y(t)) +\int_{\mathbb{R}}\frac{\partial \gamma}{\partial x}(t,z)D_1(t,z)\nu(dz)\Bigl] \\
   && + \zeta(t)\Bigl[\frac{\partial f}{\partial \pi}(t) +A_1(t)\frac{\partial b}{\partial \pi}(t) +B_1(t)\frac{\partial \sigma}{\partial \pi}(t) +B_2(t)\frac{\partial \beta}{\partial \pi}(t) \\
   && +\int_{\mathbb{R}}\frac{\partial \gamma}{\partial \pi}(t,z)D_1(t,z)\nu(dz)\Bigl]\,.
\end{eqnarray*}
Combining this and \eqref{itonecessary}, we get
\begin{equation*}
    \frac{d}{d\ell}\mathcal{J}(\pi+\ell\zeta)\Bigl|_{\ell=0}=\mathbb{E}\Bigl[\int_{\mathbb{R}} \frac{\partial \mathcal{H}}{\partial \pi}(t)\zeta(t)dt\Bigl]\,,
\end{equation*}
Then, we conclude that
$$ \frac{d}{d\ell}\mathcal{J}(\pi+\ell\zeta)\Bigl|_{\ell=0}=0 \ \ \rm{if \ and \ only \ if} \ \ \mathbb{E}\Bigl[\int_0^T \frac{\partial \mathcal{H}}{\partial \pi}(t)\zeta(t)dt\Bigl]=0\,,$$
for all bounded $\zeta\in\mathcal{A}$.

Applying this for a particular case $\zeta(t)=\xi(t)$, we get that
$$
\mathbb{E}\Bigl[ \frac{\partial \mathcal{H}}{\partial \pi}(t)\zeta(t)\mid\mathcal{F}_t\Bigl]=0\,,
$$
which completes the proof.
\begin{flushright}
$\square$
\end{flushright}

\end{document}